\documentclass[english,final]{IEEEtran}
\usepackage[T1]{fontenc}
\usepackage[latin9]{inputenc}
\usepackage{float}
\usepackage{amsmath}
\usepackage{amsthm}
\usepackage{amssymb}
\usepackage{graphicx}

\makeatletter

\floatstyle{ruled}
\newfloat{algorithm}{tbp}{loa}
\providecommand{\algorithmname}{Algorithm}
\floatname{algorithm}{\protect\algorithmname}

  \theoremstyle{plain}
  \newtheorem{lem}{\protect\lemmaname}
\theoremstyle{plain}
\newtheorem{thm}{\protect\theoremname}
  \theoremstyle{plain}
  \newtheorem{cor}{\protect\corollaryname}
  \theoremstyle{definition}
  \newtheorem{defn}{\protect\definitionname}

\usepackage{amsmath}
\usepackage{amssymb}

\usepackage{color}  % make colors
\usepackage{graphicx,psfrag,cite,subfigure}

\author{

\IEEEauthorblockN{Junting~Chen and Urbashi~Mitra}

\IEEEauthorblockA{Ming Hsieh Department of Electrical Engineering, 
University of Southern California \\
Los Angeles, CA 90089 USA, email:\{juntingc, ubli\}@usc.edu}
}


\usepackage[acronym]{glossaries}
\newcommand{\newac}{\newacronym}
\newcommand{\ac}{\gls}
\newcommand{\Ac}{\Gls}

\makeglossaries

\newac{speb}{SPEB}{square position error bound}
\newac[plural=EFIMs,firstplural=Fisher information matrices (EFIMs)]{efim}{EFIM}{Fisher information matrix}
\newac{ne}{NE}{Nash equilibrium}
\newac{mse}{MSE}{mean squared error}
\newac{toa}{TOA}{time-of-arrival}
\newac{snr}{SNR}{signal-to-noise ratio}
\newac{lan}{LAN}{local area network}
\newac{psd}{PSD}{positive semidefinite}
\newac{pd}{PD}{positive definite}
\newac{wrt}{w.r.t.}{with respect to}
\newac{lhs}{L.H.S.}{left hand side}
\newac{wp1}{w.p.1}{with probability 1}
\newac{kkt}{KKT}{Karush-Kuhn-Tucker}
\newac{wlog}{w.l.o.g.}{without loss of generality}
\newac{mle}{MLE}{maximum likelihood estimation}
\newac{gps}{GPS}{global positioning system}
\newac{rssi}{RSSI}{received signal strength}
\newac{mimo}{MIMO}{multiple-input multiple-output}
\newac{csi}{CSI}{channel state information}
\newac{fdd}{FDD}{frequency division multiplexing}
\newac{ms}{MS}{mobile station}
\newac{bs}{BS}{base station}
\newac{d2d}{D2D}{device-to-device}
\newac{slnr}{SLNR}{signal-to-interference-leakage-and-noise-ratio}
\newac{ula}{ULA}{uniform linear antenna array}
\newac{pas}{PAS}{power angular spectrum}
\newac{mmse}{MMSE}{minimum mean square error}
\newac{zf}{ZF}{zero-forcing}
\newac{rzf}{RZF}{regularized zero-forcing}
\newac{as}{AS}{angular spread}
\newac{aod}{AOD}{angle of departure}
\newac{iid}{i.i.d.}{independent and identically distributed} 
\newac{sinr}{SINR}{signal-to-interference-and-noise ratio}
\newac{tdd}{TDD}{time-division duplex}
\newac{rvq}{RVQ}{random vector quantization}
\newac{rhs}{R.H.S.}{right hand side}
\newac{mrc}{MRC}{maximum ratio combining}
\newac{cdf}{CDF}{cumulative distribution function}
\newac{a.s.}{a.s.}{almost surely}
\newac{los}{LOS}{line-of-sight}
\newac{jsdm}{JSDM}{joint spatial division and multiplexing}
\newac{map}{MAP}{maximum a posteriori}
\newac{klt}{KLT}{Karhunen-Lo\`eve Transform}
\newac{lbe}{LBE}{link bargaining equilibrium}
\newac{se}{SE}{Stackelberg equilibrium}
\newac{uav}{UAV}{unmanned aerial vehicle}
\newac{nlos}{NLOS}{non-line-of-sight}
\newac{pdf}{PDF}{probability density function}
\newac{em}{EM}{expectation-maximization}
\newac{knn}{KNN}{$k$-nearest neighbor}
\newac{svd}{SVD}{singular value decomposition}

\setkeys{Gin}{width=1.0\columnwidth}

\makeatother

\usepackage{babel}
  \providecommand{\definitionname}{Definition}
  \providecommand{\lemmaname}{Lemma}
\providecommand{\corollaryname}{Corollary}
\providecommand{\theoremname}{Theorem}

\begin{document}

\title{Rotated Eigenstructure Analysis for \\Source Localization without
Energy-decay Models}

\maketitle
% This is from Rotation_ISIT17_full_v1.lyx
%
%% The following is the formatting requirement of TWC submission
%
% 12pt, draftclsnofoot, peerreview, a4paper, oneside, onecolumn
%
% ================
%% The following commnds automatically adjust the size of the figures and the font size on the figures according to single/double column confirguration. However, it does not affect the figures with their size exiplicitly specified. It also requrie \figfontsize command to be put in the psfrag block.
%% Choice of font size: tiny, scriptsize, footnotesize, small, normalsize, large, Large, LARGE, huge Huge

%% Uncomment the following if a single column format is used.

%\newcommand*{\SINGLECOLUMN}{}

\ifdefined\SINGLECOLUMN
	% - single column setting
	\setkeys{Gin}{width=0.5\columnwidth}
	\newcommand{\figfontsize}{\footnotesize} 
\else
	% - double column setting
	\setkeys{Gin}{width=1.0\columnwidth}
	\newcommand{\figfontsize}{\normalsize} 
\fi
\begin{abstract}
Herein, the problem of simultaneous localization of two sources given
a modest number of samples is examined. In particular, the strategy
does not require knowledge of the target signatures of the sources
\emph{a priori}, nor does it exploit classical methods based on a
particular decay rate of the energy emitted from the sources as a
function of range. General structural properties of the signatures
such as unimodality are exploited. The algorithm localizes targets
based on the rotated eigenstructure of a reconstructed observation
matrix. In particular, the optimal rotation can be found by maximizing
the ratio of the dominant singular value of the observation matrix
over the nuclear norm of the optimally rotated observation matrix.
It is shown that this ratio has a unique local maximum leading to
computationally efficient search algorithms. Moreover, analytical
results are developed to show that the squared localization error
decreases at a rate $n^{-3}$ for a Gaussian field with a single source,
where $n(\log n)^{2}$ scales proportionally to the number of samples
$M$.
\end{abstract}

\section{Introduction}

\label{sec:intro}

Underwater source detection and localization is an important but challenging
problem. Classical range-based or energy-based source localization
algorithms usually require energy-decay models and the knowledge of
the environment \cite{BecStoLi:J08,QiXiuYua:J13,SheHu:J05,MeeMitNar:J08,LiuHuPan:J12,ZisWax:C88}.
However, critical environment parameters may not be available in many
underwater applications, in which case, classical model-dependent
methods may break down, even when the measurement \ac{snr} is high. 

There have been some studies on source localization using nonparametric
machine learning techniques, such as kernel regressions and support
vector machines \cite{LefRea:J17,NguJorSin:J05,JinSohWon:J10,KimParYooKimPar:J13}.
However, these methods either require a large amount of sensor data,
or some implicit information of the environment, such as the choice
of kernel functions. For example, determining the best kernel parameters
(such as bandwidth) is very difficult given a small amount of data. 

This paper focuses on source detection and localization problems when
only some structural properties of the energy field generated by the
sources are available. Specifically, instead of requiring the knowledge
of how energy decays with distance to the source, the paper aims at
exploiting only the assumption that the closer to the source the higher
energy received, and moreover, the energy field of the source is spatially
invariant and decomposable. In fact, such a structural property is
generic in many underwater applications. The prior work \cite{ChoMit:C15,ChoKumNarMit:J16}
studied the single source case, where an observation matrix is formed
from a few energy measurements of the field in the target area, and
the missing entries of the observation matrix are filled using matrix
completion methods. Knowing that the matrix would be rank-1 under
full and noise-free sampling of the whole area, \ac{svd} is applied
to extract the dominant singular vectors, and the source location
is inferred from analyzing the peaks of the singular vectors. 

Herein, we propose to improve upon two shortcomings in \cite{ChoMit:C15,ChoKumNarMit:J16}:
we make rigorous an estimation/localization bound (versus focusing
on the reduction of the search region) and we provide a method for
localizing two sources. In the two source case, we need to tackle
an additional difficulty that the \ac{svd} of the observation matrix
does not correspond to the signature vectors of the sources. To resolve
this issue, a method of rotated eigenstructure analysis is proposed,
where the observation matrix is formed by rotating the coordinate
system such that the sources are aligned in a row or in a column of
the matrix. We develop algorithms to first localize the central axis
of the two sources, and then separate the sources on the central axis. 

To summarize, we derive algorithms to simultaneously localize up to
two sources based on only a few power measurements in the target area
without knowing any specific energy-decay model. The contributions
of this paper are as follows:
\begin{itemize}
\item We derive the location estimators with analytical results to show
that the squared error decreases at a rate $n^{-3}$ for a Gaussian
field with a single source, where $n(\log n)^{2}$ scales proportionally
to the number of samples $M$. 
\item We develop a localization algorithm for the double source case based
on a novel rotated eigenstructure analysis. We show that the two sources
can be separated even when their aggregate power field has a single
peak. 
\end{itemize}

The rest of the paper is organized as follows. Section \ref{sec:system-model}
gives the system model and assumptions. Section \ref{sec:single-source}
develops location estimator with performance analysis for single source
case. Section \ref{sec:double-source} proposes rotated eigenstructure
analysis for double source case. Numerical results are given in Section
\ref{sec:numerical} and Section \ref{sec:conclusion} concludes this
work. 

\section{System Model}

\label{sec:system-model}

\begin{figure}
\subfigure[]{

\includegraphics[width=0.38\columnwidth]{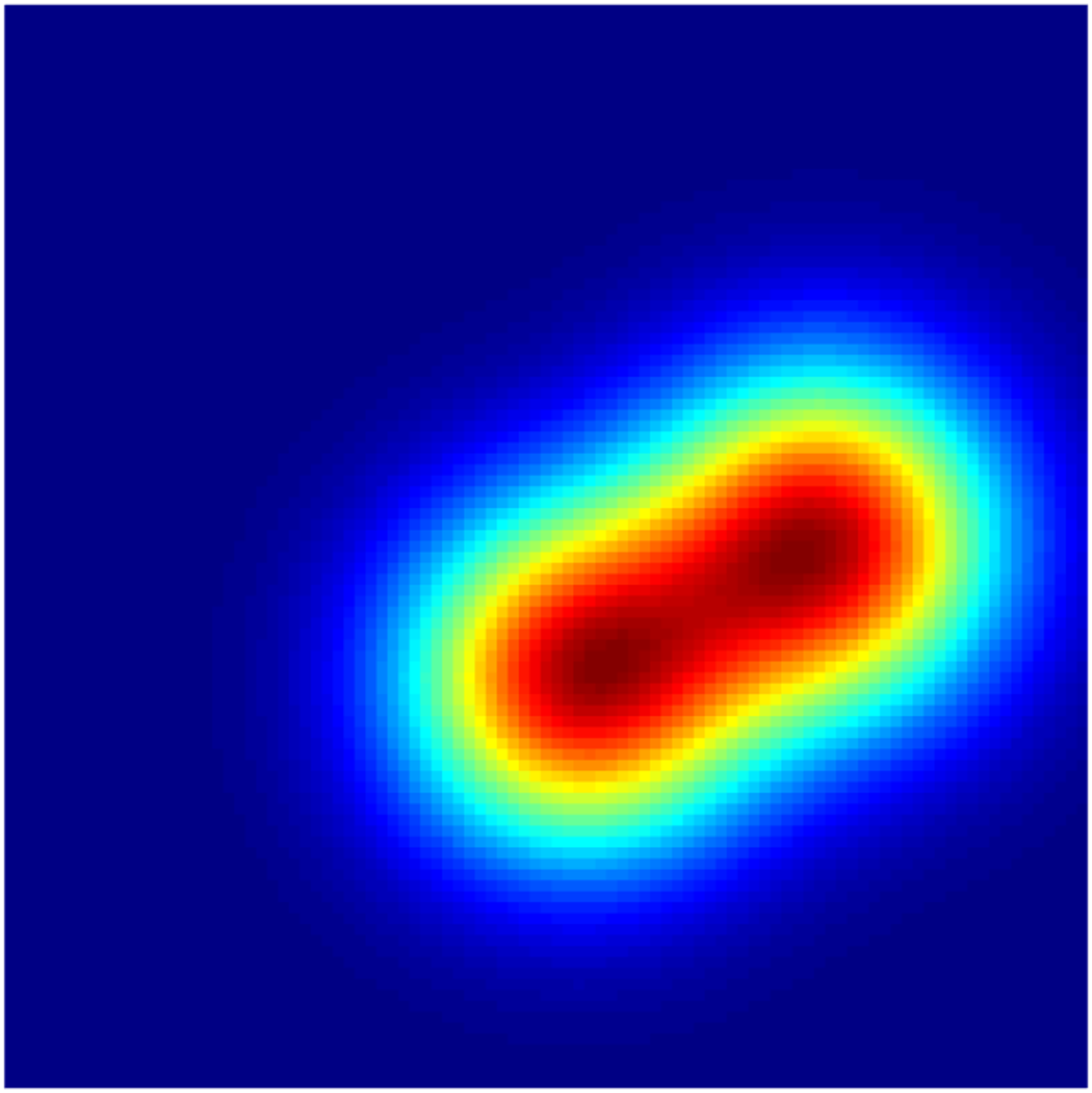}}\subfigure[]{\psfragscanon
\psfrag{0}[][][0.4]{0}
\psfrag{0.5}[][][0.4]{0.5}
\psfrag{-0.5}[][][0.4]{-0.5}
\psfrag{1}[][][0.4]{1}
\psfrag{-1}[][][0.4]{-1}
\psfrag{X-axis}[][][0.5]{X-axis}
\psfrag{Y-axis}[][][0.5]{Y-axis}
\includegraphics[width=0.55\columnwidth]{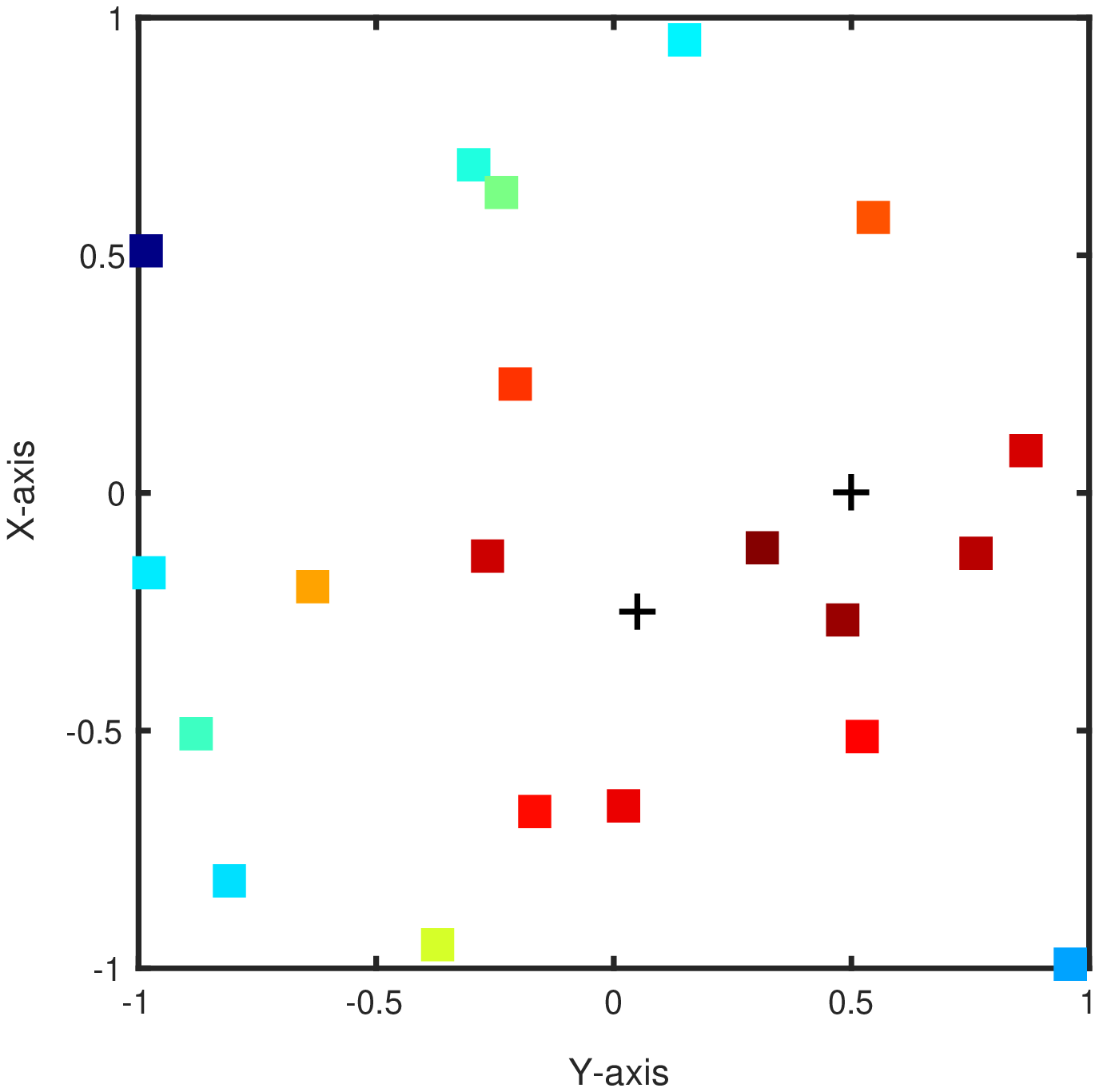}}

\caption{\label{fig:motivational-example} Without knowing the energy-decay
model, to localize the two sources in (a) based on the a small number
of measurement samples in (b), where the colored bricks represent
the sample locations and the black crosses represent the source locations.}
\end{figure}

Consider that there are $K$ ($K=1,2$) sources with unknown locations
$\mathbf{s}_{k}=(x_{k}^{\text{S}},y_{k}^{\text{S}})\in\mathbb{R}^{2}$
located in a bounded area $\mathcal{A}$. Suppose that the sensors
can only measure the aggregate power transmitted by the sources, and
is given by 
\[
h(x,y)=\sum_{k}h_{k}(x,y)
\]
for measurement location $(x,y)$, where 
\begin{equation}
h_{k}(x,y)=\alpha u(x-x_{k}^{\text{S}})u(y-y_{k}^{\text{S}})\label{eq:structure-power-density-func}
\end{equation}
is the power density from source $k$, where $\alpha>0$. The explicit
form of the density function $h_{k}(x,y)$ is unknown to the system,
except that the \emph{characteristic} function $u(x)$ is known to
have the following properties
\begin{enumerate}
\item[a)]  positive semi-definite, i.e., $u(x)\geq0$ for all $x\in\mathbb{R}$
\item[b)]  symmetric, i.e., $u(x)=u(-x)$
\item[c)]  unimodal, i.e., $u^{'}(x)<0$ for $x>0$, 
\item[d)]  smooth, i.e., $|u^{'}(x)|<K_{u}$ for some $K_{u}>0$, and 
\item[e)]  normalized, i.e., $\int_{-\infty}^{\infty}u(x)^{2}dx=1$. 
\end{enumerate}
Note that $u(x)$ can be considered as the marginal power density
function. 

Consider that $M$ power measurements $\{h^{(l)}\}$ are taken over
distinct locations $\mathbf{z}^{(l)}=(x^{(l)},y^{(l)})$, $l=1,2,\dots,M$,
uniformly at random in the target area $\mathcal{A}$. The measurements
are assigned to a $n_{1}\times n_{2}$ observation matrix $\hat{\mathbf{H}}$
as follows. First, partition the target area $\mathcal{A}$ into $n_{1}\times n_{2}$
disjoint cells $\mathcal{G}_{ij}$, $i=1,2,\dots,n_{1}$ and $j=1,2,\dots,n_{2}$,
where $n_{1}$ and $n_{2}$ are to be determined. Second, assign the
power measurements $h^{(l)}$ to the corresponding $(i,j)$th entry
of $\hat{\mathbf{H}}$ as 
\begin{equation}
\hat{H}_{ij}=s(\mathcal{G}_{ij})h^{(l)}\label{eq:sampling-model}
\end{equation}
if $\mathbf{z}^{(l)}\in\mathcal{G}_{ij}$,where $s(\mathcal{G}_{ij})$
measures the area of $\mathcal{G}_{ij}$.\footnote{If multiple samples are close to each other and assigned to the same
entry of $\hat{\mathbf{H}}$, the value of that entry is the average
of the sample values.} Denote $\Omega$ as the set of observed entries of $\hat{\mathbf{H}}$,
i.e., $(i,j)\in\Omega$ if there exists $\mathbf{z}^{(l)}\in\mathcal{G}_{ij}$
such that\textbf{ }$h^{(l)}$ is assigned to $\hat{H}_{ij}$. 

For easy discussion, assume that $\mathcal{A}=[-\frac{L}{2},\frac{L}{2}]\times[-\frac{L}{2},\frac{L}{2}]$,
$n_{1}=n_{2}=n$, and $\mathcal{G}_{ij}$ are rectangles centered
at $(x_{i},y_{j})$, $x_{i}=-\frac{L}{2}+\frac{L}{2n}+\frac{L}{n}(i-1)$,
$y_{j}=-\frac{L}{2}+\frac{L}{2n}+\frac{L}{n}(j-1)$, and have identical
size with each other. Let $\mathbf{H}=\alpha\sum_{k=1}^{K}\mathbf{u}_{k}\mathbf{v}_{k}^{\text{T}}$
be the matrix of ideal observation, where 
\begin{align}
\mathbf{u}_{k} & =\frac{L}{N}\big[u(x_{1}-x_{k}^{\text{S}}),u(x_{2}-x_{k}^{\text{S}}),\dots,u(x_{n}-x_{k}^{\text{S}})\big]^{\text{T}}\label{eq:uk}\\
\mathbf{v}_{k} & =\frac{L}{N}\big[u(y_{1}-y_{k}^{\text{S}}),u(y_{2}-y_{k}^{\text{S}}),\dots,u(y_{n}-y_{k}^{\text{S}})\big]^{\text{T}}\label{eq:vk}
\end{align}
for $k=1,2$. Thus $\mathbf{H}$ has rank at most $K$. For $(i,j)\in\Omega$,
we have $\hat{H}_{ij}\approx H_{ij}$, where the slight difference
is due to sampling away from the centers of the cells $\mathcal{G}_{ij}$.
As a result, $\hat{\mathbf{H}}$ is a sparse and noisy observation
of the low rank matrix $\mathbf{H}$. An application example is illustrated
in \ref{fig:motivational-example}.

The goal of this paper is to find the approximate locations of the
sources using only the spatial invariant property (\ref{eq:structure-power-density-func})
and the four generic properties of the characteristic function $u(x)$.
Note that this problem is non-trivial. We insist on several features
of the algorithm to be developed: it should be robust to structural
knowledge of the signatures of the sources (as captured by $g(x,y)$
in (\ref{eq:structure-power-density-func})). This disallows the use
of parametric regression or parameter estimation for source localization.
In addition, we wish to under-sample the target area using small $M$.
As such, maximum value entries may not represent the true locations
of the sources. While not a focus of the current work, we will use
matrix completion methods and the low rank property of $\mathbf{H}$
as in \cite{ChoKumNarMit:J16,ChoMit:C15} to cope with the under-sampled
observations.

\section{Eigenstructure Analysis for \\Single Source Localization}

\label{sec:single-source}

To simplify the discussion, the following mild assumptions are made.\footnote{The two assumptions are mainly to avoid discussing the effects on
the boundary of $\mathcal{A}$ and the high order noise term in the
sampling noise model (\ref{eq:sampling-noise-upper-bound}). Straight-forward
modifications can be made to handle the boundary effect in practical
algorithms.} 
\begin{enumerate}
\item[A1)]  The observation area $\mathcal{A}$ is large enough, such that there
is only negligible energy spreading outside the area $\mathcal{A}$.
\item[A2)]  The parameter $n$ is not too small, such that $u(x_{i}-x_{k}^{\text{S}})^{2}\delta^{2}\approx\int_{x_{i}}^{x_{i+1}}u(x-x_{k}^{\text{S}})^{2}dx$
and $u(y_{i}-y_{k}^{\text{S}})^{2}\delta^{2}\approx\int_{y_{i}}^{y_{i+1}}u(y-y_{k}^{\text{S}})^{2}dy$
for all $i=1,2,\dots,n$. 
\end{enumerate}
Mathematically, the above assumptions imply that the vectors $\mathbf{u}_{k}$
and $\mathbf{v}_{k}$ have unit norm.

\subsection{Observation Matrix Construction}

We first exploit the low rank property of $\mathbf{H}$ to obtain
the full matrix $\hat{\mathbf{H}}_{\text{c}}$ from the partially
observed matrix $\hat{\mathbf{H}}$. Let $\mathcal{P}_{\Omega}(\mathbf{X})$
be a projection, such that the $(i,j)$th element of matrix $\mathcal{P}_{\Omega}(\mathbf{X})$
is $\big[\mathcal{P}_{\Omega}(\mathbf{X})\big]_{ij}=X_{ij}$ if $(i,j)\in\Omega$,
and $\big[\mathcal{P}_{\Omega}(\mathbf{X})\big]_{ij}=0$ otherwise.
The completed matrix $\hat{\mathbf{H}}_{\text{c}}$ can be found as
the unique solution to the following problem 
\begin{align}
\underset{\mathbf{X}}{\text{minimize}} & \quad\|\mathbf{X}\|_{*}\label{eq:matrix-completion}\\
\text{subject to} & \quad\|\mathcal{P}_{\Omega}(\mathbf{X}-\hat{\mathbf{H}})\|_{\text{F}}\leq\epsilon\nonumber 
\end{align}
where $\|\mathbf{X}\|_{*}$ denotes the nuclear norm of $\mathbf{X}$
and $\epsilon$ is a small parameter to tolerate the discrepancy between
the two matrices. 

To choose a proper dimension $n$ for the observation matrix $\hat{\mathbf{H}}_{c}\in\mathbb{R}^{n\times n}$,
we consider the results in \cite{CanPla:J10}. It has been shown that
under some mild conditions of $\mathbf{H}$ (such as the strong incoherence
property and small rank property), the matrix $\mathbf{H}\in\mathbb{R}^{n\times n}$
can be exactly recovered with a high probability, if the dimension
$n$ satisfies $Cn(\log n)^{2}\leq M$ and noise-free sampling, $\hat{H}_{ij}=H_{ij}$
for $(i,j)\in\Omega$, is performed. Here, $C$ is a positive constant.
Given this, we propose to choose $n=n_{\text{c}}$ as the largest
integer to satisfy $n_{\text{c}}(\log n_{\text{c}})^{2}\leq M/C$.

\subsection{Location Estimator Exploiting Property of Symmetry}

Consider the \ac{svd} of the completed matrix $\hat{\mathbf{H}}_{\text{c}}$
as $\hat{\mathbf{H}}_{\text{c}}=\alpha_{1}\hat{\mathbf{u}}_{1}\hat{\mathbf{v}}_{1}^{\text{T}}+\sum_{i=2}^{n_{\text{c}}}\alpha_{i}\hat{\mathbf{u}}_{i}\hat{\mathbf{v}}_{i}^{\text{T}}$.
We thus model the singular vectors of $\hat{\mathbf{H}}_{c}$ as $\hat{\mathbf{u}}_{1}=\mathbf{u}_{1}+\mathbf{e}_{u}$
and $\hat{\mathbf{v}}_{1}=\mathbf{v}_{1}+\mathbf{e}_{v}$.

Note that the vectors $\mathbf{u}_{1}$ and $\mathbf{v}_{1}$ defined
in (\ref{eq:uk}) and (\ref{eq:vk}), respectively, contain the source
location information due to the unimodal property of $u(x)$. However,
due to the noise vectors $\mathbf{e}_{u}$ and $\mathbf{e}_{v}$,
the source location cannot be found by simply locating the peaks of
$\hat{\mathbf{u}}_{1}$ and $\hat{\mathbf{v}}_{1}$. 

To resolve this difficulty, we exploit the symmetric property of $u(x)$
and develop a location estimator as follows.

Define a reflected correlation function as 
\begin{equation}
\hat{R}(t;\hat{\mathbf{u}}_{1})=\int_{-\infty}^{\infty}\hat{u}(x)\hat{u}(-x+t)dx\label{eq:reflected-correlation}
\end{equation}
where $\hat{u}(x)$ is a (nonparametric) regression function from
vector $\hat{\mathbf{u}}_{1}$. For example, $\hat{u}(x)$ can be
obtained by $\hat{u}(x)=\hat{\mathbf{u}}_{1}(i)$ if $x=x_{i}$, and
by linear interpolation between $\hat{\mathbf{u}}_{1}(i)$ and $\hat{\mathbf{u}}_{1}(i+1)$
if $x_{i}<x<x_{i+1}$. Then the location estimator for $x_{1}^{\text{S}}$
is given by 
\begin{equation}
\hat{x}_{1}^{\text{S}}(\hat{\mathbf{u}}_{1})=\frac{1}{2}\underset{t\in\mathbb{R}}{\text{argmax}}\,\hat{R}(t;\hat{\mathbf{u}}_{1}).\label{eq:location-estimator}
\end{equation}
The location estimator for $y_{1}^{\text{S}}$ can be obtained in
a similar way.

The location estimator (\ref{eq:location-estimator}) exploits the
fact that as $\hat{\mathbf{u}}_{1}$ is symmetric, the reflected correlation
(\ref{eq:reflected-correlation}), which is the correlation between
$\hat{\mathbf{u}}_{1}$ and a reflected and shifted version of $\hat{\mathbf{u}}_{1}$,
is maximized at the source location. Therefore, the estimator $\hat{x}_{1}^{\text{S}}(\hat{\mathbf{u}}_{1})$
tries to the suppress the perturbation from the noise by correlating
over all the entries of $\hat{\mathbf{u}}_{1}$. 

We establish several properties for the estimator $\hat{x}_{1}^{\text{S}}(\hat{\mathbf{u}}_{1})$. 

Consider the autocorrelation for the characteristic function $u(x)$
as 
\begin{equation}
\tau(t)=\int_{-\infty}^{\infty}u(x)u(x-t)dx.\label{eq:autocorrelation-function}
\end{equation}
Then, the following property can be derived. 
\begin{lem}
[Monotonicity] \label{lem:Monotonicity} The autocorrelation function
$\tau(t)$ is non-negative and symmetric. In addition, $\tau(t)$
is strictly decreasing in $t>0$.
\end{lem}

Let the dominant singular vector of $\hat{\mathbf{H}}_{\text{c}}$
as the solution to (\ref{eq:matrix-completion}) be given by $\hat{\mathbf{u}}_{1}=\mathbf{u}_{1}+\mathbf{e}_{1}$,
where $\mathbf{u}_{1}$ is the dominant singular of $\mathbf{H}$.
Let $\overleftarrow{\mathbf{e}}_{1}$ be a vector with reverse elements
of $\mathbf{e}_{1}$, i.e., the $j$th element of $\overleftarrow{\mathbf{e}}_{1}$
equals to the last but the $j$th element of $\mathbf{e}_{1}$. Let
$\mathbf{e}_{1}^{-t}$ be a vector obtained from the $t$-shift of
$\overleftarrow{\mathbf{e}}_{1}$, i.e., for $t>0$, the first $t$
elements of $\mathbf{e}_{1}^{-t}$ are zeros and the remaining $(n_{c}-t)$
elements of $\mathbf{e}_{1}^{-t}$ are identical to the first $(n_{c}-t)$
elements of $\overleftarrow{\mathbf{e}}_{1}$; and for $t<0$, the
first $(n_{c}-t)$ elements of $\mathbf{e}_{1}^{-t}$ are identical
to the last $(n_{c}-t)$ elements of $\overleftarrow{\mathbf{e}}_{1}$
and the remaining $t$ elements of $\mathbf{e}_{1}^{-t}$ are zeros.
With such a notion, we make the following assumption on the singular
vector $\hat{\mathbf{u}}_{1}=\mathbf{u}_{1}+\mathbf{e}_{1}$ of the
completed matrix $\hat{\mathbf{H}}_{\text{c}}$: 
\begin{equation}
|\mathbf{u}_{1}^{\text{T}}\mathbf{e}_{1}^{-t}|\leq C_{e}|\mathbf{u}_{1}^{\text{T}}\mathbf{e}_{1}|\label{eq:approximation-ue}
\end{equation}
for any $0\leq t\leq n_{c}-1$, where $C_{e}<\infty$ is a positive
constant that only depends on the characteristic function $u(x)$
but not $n_{c}$ or $M$. 

Such an approximation is motivated by two observations. First, the
entries of the vector $\mathbf{e}_{1}$ may have roughly the same
chance to take positive values or negative values because both $\mathbf{u}_{1}$
and $\mathbf{u}_{1}+\mathbf{e}_{1}$ have unit norm. Second, the magnitude
of the elements in $\mathbf{u}_{1}$ depends only on the characteristic
function $u(x)$ but not $n_{c}$ or $M$. Although it is difficult
to analytically validate the assumption (\ref{eq:approximation-ue}),
it can be roughly confirmed by massive simulation results. 

As a result, we have the following theorem to characterize the estimation
error of $\hat{x}_{1}^{\text{S}}$.
\begin{thm}
[Localization error bound]\label{thm:Localization-error-bound}
Suppose that the sampling error of $\hat{\mathbf{H}}$ from the true
energy field matrix $\mathbf{H}$ is bounded by $\|\mathcal{P}_{\Omega}(\hat{\mathbf{H}}-\mathbf{H})\|_{\text{F}}\leq\bar{\epsilon}$
and the algorithm parameter $\epsilon$ in (\ref{eq:matrix-completion})
is chosen as $\epsilon=\bar{\epsilon}$. Then, with high probability,
\begin{equation}
|\hat{x}_{1}^{\text{S}}-x_{1}^{\text{S}}|\leq\frac{1}{2}\tau^{-1}\big(1-\mu_{u}L^{6}n_{c}(M)^{-3}+o(n_{c}(M)^{-3})\big)\label{eq:localization-bound}
\end{equation}
where $\tau^{-1}(r)$ is the inverse function of $r=\tau(t)$, $\mu_{u}=C_{e}128u(0)^{2}K_{u}^{2}$,
and $n_{c}(M)$ is the largest integer chosen such that $M\geq Cn_{c}(\log n_{c})^{2}$.
\end{thm}

The specific performance from (\ref{eq:localization-bound}) depends
on the characteristics of the energy field. Intuitively, if $u(x)$
has a sharp peak (large slope of the autocorrelation function $\tau(t)$),
the localization error should be smaller. Consider a numerical example
where the energy field has a Gaussian characteristic function.
\begin{cor}
[Squared error bound in Gaussian field]\label{cor:squared-error-Gaussian}
For a Gaussian characteristic function $u(x)=\big(\frac{2\gamma}{\pi}\big)^{\frac{1}{4}}e^{-\gamma x^{2}}$,
there exists a constant $C_{\mu}$, which only depends on the characteristic
function $u(x)$, such that with high probability, the squared estimation
error is upper bounded by 
\begin{equation}
|\hat{x}_{1}^{\text{S}}-x_{1}^{\text{S}}|^{2}+|\hat{y}_{1}^{\text{S}}-y_{1}^{\text{S}}|^{2}\leq C_{\mu}L^{6}n_{c}(M)^{-3}+o(n_{c}(M)^{-3}).\label{eq:sqaured-error-bound-Gaussian}
\end{equation}
\end{cor}

Theorem \ref{thm:Localization-error-bound} and Corollary \ref{cor:squared-error-Gaussian}
gives the asymptotic performance of the proposed localization algorithm
without knowing the energy-decay model. For large $M$, the worst
case squared error decays at a rate $n_{c}(M)^{-3}$. As a benchmark,
the squared error of a naive scheme, which estimates the source location
directly from the position of the measurement sample that observes
the highest power, decreases as $M^{-1}$, which is equivalent to
$n_{c}(M)^{-1}(\log n_{c}(M))^{-2}$, much slower than that of the
proposed algorithm. This is because, the granularity of the original
observations is $L/\sqrt{M}$. The results then confirm that by exploiting
the low rank property using matrix completion and the reflected correlation
technique, the proposed algorithm significantly improves the localization
resolution.

\section{Rotated Eigenstructure Analysis for \\Double Source Localization}

\label{sec:double-source}

The location estimator $\hat{x}_{1}^{\text{S}}$ in (\ref{eq:location-estimator})
is based on the intuition that the singular vectors of $\mathbf{H}$
are just the vectors $\mathbf{u}_{1}$ and $\mathbf{v}_{1}$, which
contains the source location in their peaks. However, a similar technique
cannot be applied to the two source case, because $\mathbf{u}_{k}$
and $\mathbf{v}_{k}$ may not be the singular vectors of $\mathbf{H}$,
as the vectors $\{\mathbf{u}_{k}\}$ may not be orthogonal.

\subsection{Optimal Rotation of the Observation Matrix}

When there are two sources, the (ideal) observation matrix $\mathbf{H}$
is not rank-1, expect for the special case where the two sources are
aligned on one of the axes of the coordinate system. 

\Ac{wlog}, assume that the sources are aligned with the $x$-axis,
where $y_{k}^{\text{S}}=C$ for $k=1,2$. Consequently, we have $\mathbf{v}_{1}=\mathbf{v}_{2}$,
and $\mathbf{H}=\alpha\big(\sum_{k}\mathbf{u}_{k}\big)\mathbf{v}_{1}^{\text{T}}$,
which is rank-1. Hence, the right singular vector of $\mathbf{H}$
is $\mathbf{v}_{1}$ and, by analyzing the peak of $\mathbf{v}_{1}$,
the central axis $\hat{y}_{k}^{\text{S}}=C$ can be estimated. 

The above observations suggest that we rotate the coordinate system
such that the sources are aligned with one of the axes. Consider rotating
the coordinate system by $\theta$. The entries of $\hat{\mathbf{H}}_{\text{c}}$
are rearranged into a new observation matrix $\hat{\mathbf{H}}_{\theta}$,
where 
\begin{equation}
\big[\hat{\mathbf{H}}_{\theta}\big]_{(i,j)}=\big[\hat{\mathbf{H}}_{\text{c}}\big]_{(p,q)}\label{eq:H-theta}
\end{equation}
in which $(p,q)$ is the index such that $(x_{p}^{'},y_{q}^{'})$
is the closest point in Euclidean distance to $(\bar{x},\bar{y})$
in the original coordinate system $\mathcal{C}_{0}$, with $\bar{x}=d\cos(\beta+\theta)$
and $\bar{y}=d\sin(\beta+\theta)$. Here $\beta=\angle(x_{i},y_{j})$
is the angle of $(x_{i},y_{j})$ to the $x$-axis of the rotated coordinate
system $\mathcal{C}_{\theta}$, and $d=\|(x_{i},y_{j})\|_{2}$. Note
that $1\leq i,j\leq n^{'}$, where $n^{'}\leq n_{c}$, since the rotation
of the axes induce truncation of some data samples. 

Let the orientation angle of the central axis of the sources \ac{wrt}
the $x$-axis in the original coordinate system $\mathcal{C}_{0}$
be $\theta_{0}$, $\theta_{0}\in[0,\pi)$. Then the desired rotation
for coordinate system $\mathcal{C}_{\theta}$ would be $\theta^{*}=\theta_{0}$
for $\theta_{0}<\frac{\pi}{2}$, or $\theta^{*}=\theta_{0}-\frac{\pi}{2}$
for $\theta_{0}\geq\frac{\pi}{2}$. The desired rotation $\theta$
can be obtained as 
\begin{equation}
\underset{\theta\in[0,\frac{\pi}{2}]}{\text{maximize}}\quad\rho(\theta)\triangleq\frac{\lambda_{1}(\hat{\mathbf{H}}_{\theta})}{\sum_{k}\lambda_{k}(\hat{\mathbf{H}}_{\theta})}\label{eq:rho-function}
\end{equation}
where $\lambda_{k}(\mathbf{A})$ is the $k$th largest singular value
of $\mathbf{A}$. Note that $\rho(\theta)\leq1$ for all $\theta\in[0,\frac{\pi}{2}]$
and $\rho(\theta^{*})=1$, where $\hat{\mathbf{H}}_{\theta}$ becomes
a rank-1 matrix when the sources are aligned with one of the axes.

The maximization problem (\ref{eq:rho-function}) is in general non-convex.
An exhaustive search for the solution $\theta^{*}$ is computationally
expensive, since for each $\theta$, \ac{svd} should be performed
to obtain the singular value profile of $\hat{\mathbf{H}}_{\theta}$.
Therefore, we need to study the properties of the alignment metric
$\rho(\theta)$ in order to develop efficient algorithms for the source
detection.

\subsection{The Unimodal Property}

We also show that the function $\rho(\theta)$ also has the unimodal
property defined as follows. 
\begin{defn}
[Unimodality] A function $f(x)$ is called unimodal in a bounded
region $(a,b)$, if there exists $x_{0}\in[a,b]$, such that $f^{'}(x)f^{'}(y)<0$
for any $a<x<x_{0}<y<b$. 
\end{defn}

The unimodality suggests that $f(x)$ has a single peak in $(a,b)$,
and hence $f(x)$ has a unique local maximum (or minimum). 
\begin{thm}
[Unimodality in the two source case]\label{thm:Unique-local-maximum}
The function $\rho(\theta)$ in (\ref{eq:rho-function}) is unimodal
in $\theta\in(\theta^{*}-\frac{\pi}{4},\theta^{*}+\frac{\pi}{4})$,
if 
\begin{equation}
s\cdot\tau^{'}(t)>t\cdot\tau^{'}(s)\label{eq:correlation-condition}
\end{equation}
for all $0<s<t$, where $\tau^{'}(t)\triangleq\frac{d}{dt}\tau(t)$.
In addition, $\rho(\theta)$ is strictly increasing over $(\theta^{*}-\frac{\pi}{4},\theta^{*})$
and strictly decreasing over $(\theta^{*},\theta^{*}+\frac{\pi}{4})$. 
\end{thm}

The result in Theorem \ref{thm:Unique-local-maximum} is powerful,
since it confirms that the function $\rho(\theta)$ is unimodal within
a $\frac{\pi}{2}$-window, and there is a unique local maximum, when
the autocorrelation of the energy field characteristic function $u(x)$
agrees with the condition (\ref{eq:correlation-condition}). Note
that $\rho(\theta)$ is also symmetric \ac{wrt} $\theta=\theta^{*}$.
As a result, a simple bisection search algorithm can efficiently find
the global optimal solution $\theta^{*}$ to (\ref{eq:rho-function}).
An example algorithm is given in Algorithm \ref{alg:optimal-rotation-angle}.

\begin{algorithm}
\begin{enumerate}
\item Let $\theta_{\text{L}}=0$ and $\theta_{\text{R}}=\frac{\pi}{2}$.
Choose an integer $T\geq1$ for smoothing (for sampling noise tolerance). 
\item \label{enu:alg-loop-start}Let $\theta_{c}=\frac{1}{2}(\theta_{\text{L}}+\theta_{\text{R}})$.
Take uniformly $T$ points in $[\theta_{\text{L}},\theta_{c}]$, i.e.,
$\theta_{i}=\theta_{c}-\frac{i}{T}(\theta_{c}-\theta_{\text{L}})$,
and compute $\bar{\rho}_{\text{L}}(\theta_{\text{L}},\theta_{\text{R}})=\frac{1}{T}\sum_{i=1}^{T}\rho(\theta_{i})$
using (\ref{eq:H-theta}) and (\ref{eq:rho-function}). Compute $\bar{\rho}_{\text{R}}(\theta_{\text{L}},\theta_{\text{R}})$
in the similar way. 
\item If $\bar{\rho}_{\text{L}}>\bar{\rho}_{\text{R}}$, then $\theta_{\text{R}}=\theta_{c}$;
otherwise, $\theta_{\text{L}}=\theta_{c}$.
\item Repeat form Step \ref{enu:alg-loop-start}) until $\theta_{\text{R}}-\theta_{\text{L}}$
small enough. Then $\theta^{*}=\theta_{c}$ is found. 
\end{enumerate}
\caption{\label{alg:optimal-rotation-angle}Search for the optimal rotation
angle}
\end{algorithm}

Note that condition (\ref{eq:correlation-condition}) can be satisfied
by a variety of energy fields. For example, for Laplacian field $u(x)=\sqrt{\gamma}e^{-\gamma|x|}$,
we have $\tau(t)=(1+\gamma t)e^{-\gamma t}$, and $\tau^{'}(t)=-\gamma^{2}te^{-\gamma t}$;
for Gaussian field $u(x)=\big(\frac{2\gamma}{\pi}\big)^{\frac{1}{4}}e^{-\gamma x^{2}}$,
we have $\tau(t)=e^{-\gamma t^{2}/2}$, and $\tau^{'}(t)=-\gamma te^{-\gamma t^{2}/2}$.
In both cases, condition (\ref{eq:correlation-condition}) is satisfied. 

% to be fixed

%

%

\subsection{Source Detection}

%

% large and small separation for detection

In the coordinate system $\mathcal{C}_{\theta}$ under optimal rotation
$\theta=\theta^{*}$ (assuming alignment on the $x$-axis), the left
and right singular vectors of $\hat{\mathbf{H}}_{\theta}$ can be
modeled as $\hat{\mathbf{u}}_{1}=\frac{1}{2}(\mathbf{u}_{1}(\theta^{*})+\mathbf{u}_{2}(\theta^{*}))+\mathbf{e}_{u}$
and $\hat{\mathbf{v}}_{1}=\mathbf{v}_{1}(\theta^{*})+\mathbf{e}_{v}$,
respectively. Correspondingly, the $y$-coordinates of the sources
can be the found using estimator (\ref{eq:location-estimator}) based
on reflected correlation 
\begin{equation}
\hat{y}_{1}^{\text{S}}(\hat{\mathbf{v}}_{1};\theta^{*})=\hat{y}_{2}^{\text{S}}(\hat{\mathbf{v}}_{1};\theta^{*})=\frac{1}{2}\underset{t\in\mathbb{R}}{\text{argmax}}\,\hat{R}(t;\hat{\mathbf{v}}_{1}).\label{eq:location-estimator-two-source-y}
\end{equation}

To find the $x$-coordinates, note that the function $u_{1}(x)=\frac{1}{2}\big(u(x-x_{1}^{\text{S}})+u(x-x_{2}^{\text{S}})\big)$
is symmetric at $x=\frac{1}{2}(x_{1}^{\text{S}}+x_{2}^{\text{S}})$.
Therefore, the center of the two sources can be found by 
\begin{equation}
\hat{c}=\frac{1}{2}\underset{t\in\mathbb{R}}{\text{argmax}}\,\hat{R}(t;\hat{\mathbf{u}}_{1}).\label{eq:location-estimator-two-source-center}
\end{equation}

In addition, after estimating $\hat{y}_{1}^{\text{S}}$, the marginal
power density function $u(x)$ can be obtained as $\hat{u}(y)=\hat{v}_{1}(y-\hat{y}_{1}^{\text{S}})$,
where $\hat{v}_{1}(y)$ is a regression function from $\hat{\mathbf{v}}_{1}$
(for example, by linear interpolation among $y_{1},y_{2},\dots,y_{n_{c}}$).
As a results, the $x$-coordinates of the two sources can be found
using similar techniques as spread spectrum early gate synchronization
\cite{peterson:b95}, and obtained as $\hat{x}_{1}^{\text{S}}(\theta^{*})=\hat{c}-\hat{d}$
and $\hat{x}_{2}^{\text{S}}(\theta^{*})=\hat{c}+\hat{d}$, where 
\begin{align}
\hat{d} & =\underset{d\geq0}{\text{argmax}}\quad Q(d;\hat{\mathbf{u}}_{1},\hat{\mathbf{v}}_{1})\label{eq:location-estimator-two-source-d}
\end{align}
and 
\[
Q(d;\hat{\mathbf{u}}_{1},\hat{\mathbf{v}}_{1})\triangleq\frac{1}{2}\int_{-\infty}^{\infty}\hat{u}_{1}(x)\Big(\hat{u}(x-\hat{c}-d)+\hat{u}(x-\hat{c}+d)\Big)dx.
\]
It is straight-forward to show that $Q(d;\hat{\mathbf{u}}_{1},\hat{\mathbf{v}}_{1})$
is maximized at $d^{*}=\frac{1}{2}|x_{1}^{\text{S}}-x_{2}^{\text{S}}|$.

As a benchmark, consider a naive scheme that estimates $x_{1}^{\text{S}}$
and $x_{2}^{\text{S}}$ by analyzing the peaks of $\hat{\mathbf{u}}_{1}$.
However, such naive strategy cannot work for small source separation,
because if $d=\frac{1}{2}|x_{1}^{\text{S}}-x_{2}^{\text{S}}|$ is
too small, the aggregate power density function $\tilde{u}_{1}(x)=u(x-x_{1}^{\text{S}})+u(x-x_{1}^{\text{S}}-d)$
would be unimodal and there is only one peak in $\hat{\mathbf{u}}_{1}$.
As a comparison, the proposed procedure estimator from procedure (\ref{eq:location-estimator-two-source-y})--(\ref{eq:location-estimator-two-source-d})
does not such a limitation.

\section{Numerical Results}

\label{sec:numerical}

\begin{figure}
\begin{centering}
\subfigure[One source]{
\psfragscanon
\psfrag{MSE}[][][0.7]{MSE}
\psfrag{M samples}[][][0.7]{Number of samples $M$}
\psfrag{20}[][][0.7]{20}
\psfrag{200}[][][0.7]{200}
\psfrag{2000}[][][0.7]{2000}
\psfrag{Naive}[Bl][Bl][0.65]{Naive}
\psfrag{Asymptotic decay rate}[Bl][Bl][0.65]{Worst case bound}
\psfrag{Proposed}[Bl][Bl][0.65]{Proposed}\includegraphics[width=1\columnwidth]{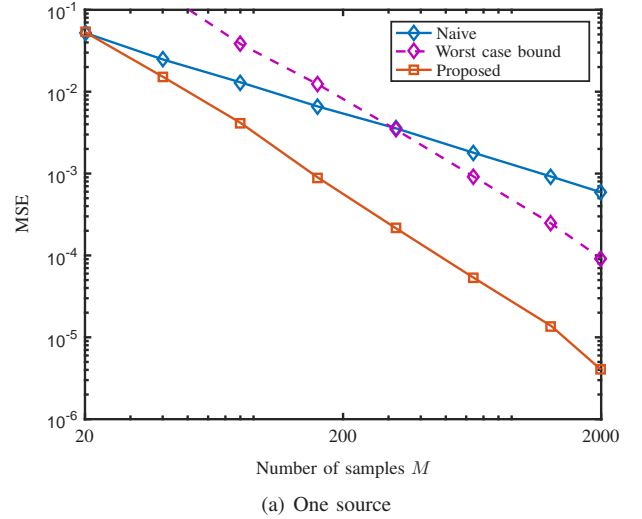}}
\par\end{centering}
\begin{centering}
\subfigure[Two sources]{
\psfragscanon
\psfrag{MSE}[][][0.7]{MSE}
\psfrag{M samples}[][][0.7]{Number of samples $M$}
\psfrag{20}[][][0.7]{20}
\psfrag{200}[][][0.7]{200}
\psfrag{2000}[][][0.7]{2000}
\psfrag{Naive}[Bl][Bl][0.65]{Naive}
\psfrag{Asymptotic decay rate}[Bl][Bl][0.65]{Worst case bound}
\psfrag{Proposed}[Bl][Bl][0.65]{Proposed}\includegraphics[width=1\columnwidth]{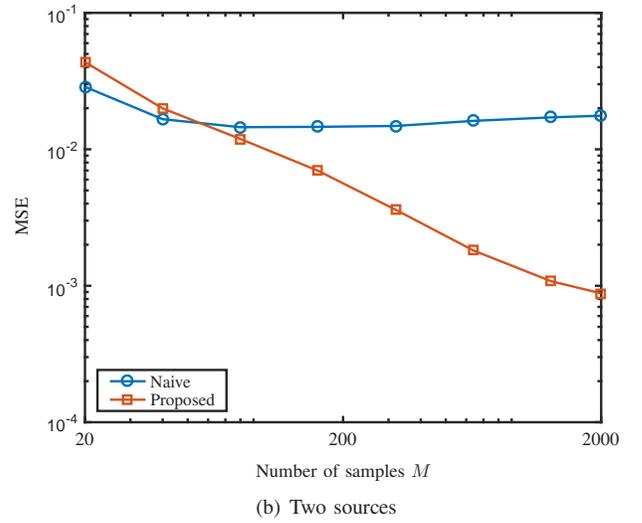}}
\par\end{centering}
\caption{\label{fig:mse} \ac{mse} of the source location versus the number
of samples $M$.}
\end{figure}

In this section, we evaluate the performance of the proposed location
estimator in both single source and double source cases. Two sources
are placed in the area $[-0.5,0.5]\times[-0.5,0.5]$ uniformly and
independently at random, with the restriction that the distance between
the two sources is no more than $0.5$.\footnote{When the two sources are far apart, the problem degenerates to two
single-source-localization problems.} The power field generated by each source in an underwater environment
is modeled as $h_{k}(x,y)=e^{-20(x-x_{k}^{\text{S}})^{2}-20(x-y_{k}^{\text{S}})^{2}}$,
$k=1,2$. There are $M$ power measurements taken in the area $\mathcal{A}=[-1,1]\times[-1,1]$
uniformly at random. The parameter $n_{c}$ of the proposed observation
matrix $\hat{\mathbf{H}}\in\mathbb{R}^{n_{c}\times n_{c}}$ is chosen
as the largest integer satisfying $n_{c}(\log n_{c})^{2}\leq M/C$,
for $C=1$. 

As a benchmark, the proposed location estimation is compared with
the naive scheme, which determines the source location directly form
the position of the measurement sample that observes the highest power.
In the two source case, the naive algorithm aims at detecting either
one of the sources, and the corresponding localization error is computed
as $\mathcal{E}_{\text{\scriptsize naive}}^{2}=\min\{\|\hat{\mathbf{s}}_{\text{\scriptsize naive}}-\mathbf{s}_{1}\|^{2},\|\hat{\mathbf{s}}_{\text{\scriptsize naive}}-\mathbf{s}_{2}\|^{2}\}$.
As a comparison, the localization error of the proposed algorithm
is computed as $\mathcal{E}^{2}=\frac{1}{2}(\|\hat{\mathbf{s}}_{1}-\mathbf{s}_{1}\|^{2}+\|\hat{\mathbf{s}}_{2}-\mathbf{s}_{2}\|^{2})$.

\begin{figure}
\begin{centering}
\psfragscanon
\psfrag{X axis}[][][0.7]{X axis}
\psfrag{Y axis}[][][0.7]{Y axis}
\psfrag{0}[][][0.7]{0}
\psfrag{1}[][][0.7]{1}
\psfrag{-1}[][][0.7]{-1}
\psfrag{0.5}[][][0.7]{0.5}
\psfrag{-0.5}[][][0.7]{-0.5}
\includegraphics[width=1\columnwidth]{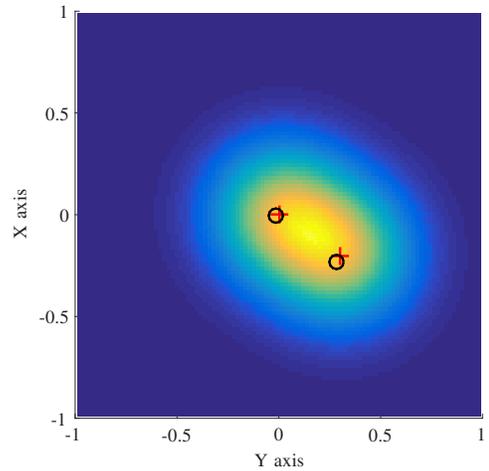}
\par\end{centering}
\caption{\label{fig:map} Localizing two sources using $M=200$ samples, where
red crosses denote the true source locations, and black circles denote
the estimates. The color map represents the aggregate power field
generated by the two sources.}
\end{figure}

Fig.~\ref{fig:mse} depicts the \ac{mse} of the source location
versus the number of samples $M$. In the single source case,the coefficient
of the worst case upper bound (\ref{eq:sqaured-error-bound-Gaussian})
is chosen as $C_{\mu}=1$ to demonstrate the asymptotic decay rate
of the worst case squared error bound. The decay rate of the analytic
worst case error bound is roughly the same as the \ac{mse} obtained
from the numerical experiment. It is expected that as $M$ increases,
the two curves merge in an asymptotic way. As a benchmark, the proposed
scheme requires less than half of the samples to achieve similar performance
to that of the naive baseline even for small $M$ (around $50$).
More importantly, it demonstrates a higher \ac{mse} decay rate, where
for medium $M$ (around $200$), the proposed scheme reduces the number
of samples to $1/10$. In the double source case, there is an error
floor for the naive scheme, because the location that observes the
highest power may not be either one of the source locations. As a
comparison, there is no error floor in for proposed scheme as $M$
increases. 

Fig.~\ref{fig:map} shows an example on simultaneously localizing two
sources (red crosses). Although the aggregate power field has only
one peak, the algorithm (black circles) is able to separate the two
sources. 

\section{Conclusions}

\label{sec:conclusion}

This paper developed source localization algorithms from a few power
measurement samples, while no specific energy-decay model is assumed.
Instead, the proposed method only exploited the structural property
of the power field generated by the sources. Analytical results were
developed to demonstrate that the proposed algorithm decreases the
localization error at a higher rate than the baseline algorithm when
the number of samples increases. In addition, a rotated eigenstructure
analysis technique was derived for simultaneously localizing two sources.
Numerical results demonstrate the performance advantage in localizing
single or double sources.

\section*{Acknowledgments}

This research was supported, in part, by National Science Foundation
under Grant NSF CNS-1213128, CCF-1410009, CPS-1446901, Grant ONR N00014-15-1-2550,
and Grant AFOSR FA9550-12-1-0215.

\section*{Appendix}

\subsection{Proof of Lemma \ref{lem:Monotonicity}}

\begin{align}
\tau^{'}(t) & =\frac{d}{dt}\int_{-\infty}^{\infty}u(x)u(x-t)dx\nonumber \\
 & =\int_{-\infty}^{\infty}-u(x)u^{'}(x-t)dx\nonumber \\
 & =-\int_{-\infty}^{0}u(z+t)u^{'}(z)dz-\int_{0}^{\infty}u(z+t)u^{'}(z)dz\nonumber \\
 & =-\int_{-\infty}^{0}u(z+t)u^{'}(z)dz+\int_{0}^{\infty}u(z+t)u^{'}(-z)dz\label{eq:app-lem-tau-eq1}\\
 & =-\int_{-\infty}^{0}u(z+t)u^{'}(z)dz+\int_{-\infty}^{0}u(-w+t)u^{'}(w)dw\label{eq:app-lem-tau-eq2}\\
 & =-\int_{-\infty}^{0}\big[u(z+t)-u(-z+t)\big]u^{'}(z)dz\nonumber \\
 & =-\int_{-\infty}^{0}\big[u(z+t)-u(z-t)\big]u^{'}(z)dz\label{eq:app-lem-tau-eq3}\\
 & <0\nonumber 
\end{align}
where (\ref{eq:app-lem-tau-eq1}) is due to the change of variable
$z=x-t$ and $u^{'}(z)=-u^{'}(-z)$, (\ref{eq:app-lem-tau-eq2}) is
to change the variable $z=-w$, (\ref{eq:app-lem-tau-eq3}) exploits
the fact that $u(x)=u(-x)$, and the last inequality is due to $u(z+t)-u(z-t)>0$
and $u^{'}(z)>0$ for all $z<0$.

\subsection{Proof of Theorem \ref{thm:Localization-error-bound}}

To simplify the algebra, we only focus on the dominant terms \ac{wrt}
$n_{c}$ as $n_{c}$ goes large. 

\subsubsection{Upper Bound of the Sampling Error}

For notational convenience, define $u_{1}(x)=u(x-x_{1}^{\text{S}})$
and $v_{1}(y)=u(x-y_{1}^{\text{S}})$. Consider the sampling position
$(x,y)\in\mathcal{G}_{ij}$. Using a Taylor expansion, we have 
\begin{align*}
 & |h_{1}(x,y)-h_{1}(x_{1},y_{1})|\\
 & \quad=\alpha|u_{1}(x)v_{1}(y)-u_{1}(x_{1})v_{1}(y_{1})|\\
 & \quad=\alpha\big|\big(u_{1}(x_{1})+u_{1}^{'}(x_{1})(x-x_{1})\big)\\
 & \qquad\qquad\times\big(v_{1}(y_{1})+v_{1}^{'}(y_{1})(y-y_{1})\big)-u_{1}(x_{1})v_{1}(y_{1})\\
 & \qquad\qquad\qquad\qquad\qquad\qquad+o(x-x_{1})+o(y-y_{1})\big|\\
 & =\alpha\big|u_{1}(x_{1})v_{1}^{'}(y_{1})(y-y_{1})+v_{1}(y_{1})u_{1}^{'}(x_{1})(x-x_{1})\big|\\
 & \qquad\qquad\qquad\qquad\qquad\qquad+o(x-x_{1})+o(y-y_{1})\\
 & \leq\alpha u(0)K_{u}\frac{L}{n_{c}}+o\Big(\frac{L}{n_{c}}\Big)
\end{align*}
from the property $u(x)\leq u(0)$ and $|u^{'}(x)|\leq K_{u}$. 

From (\ref{eq:sampling-model}), we have 
\begin{align}
|\hat{H}_{ij}-H_{ij}| & =\bigg(\frac{L}{n_{c}}\bigg)^{2}|h_{1}(x,y)-h_{1}(x_{1},y_{1})|\nonumber \\
 & \leq\alpha u(0)K_{u}\frac{L^{3}}{n_{c}^{3}}+o\Big(\frac{L^{3}}{n_{c}^{3}}\Big).\label{eq:sampling-noise-upper-bound}
\end{align}
As a result, 
\begin{align*}
\|\mathcal{P}_{\Omega}(\hat{\mathbf{H}}_{\text{c}}-\mathbf{H})\|_{\text{F}}^{2} & =\sum_{(i,j)\in\Omega}|\hat{H}_{ij}-H_{ij}|^{2}\\
 & \leq M\left(\alpha u(0)K_{u}L^{3}/n_{c}^{3}\right)^{2}\triangleq\bar{\epsilon}^{2}.
\end{align*}

\subsubsection{Matrix Completion with Noise and Singular Vector Perturbation}

When there is sampling noise, the performance of matrix completion
can be evaluated by the following result. 
\begin{lem}
[Matrix completion with noise \cite{CanPla:J10}]\label{prop:matrix-completion-noise}
Consider that $\epsilon$ in (\ref{eq:matrix-completion}) is chosen
such that $\|\mathcal{P}_{\Omega}(\hat{\mathbf{H}}-\mathbf{H})\|_{\text{F}}\leq\epsilon=\bar{\epsilon}$.
Then, with high probability, 
\begin{equation}
\delta\triangleq\|\hat{\mathbf{H}}_{\text{c}}-\mathbf{H}\|_{\text{F}}\leq4\sqrt{\frac{(2+p)n_{c}}{p}}\bar{\epsilon}+2\bar{\epsilon}\label{eq:matrix-completion-noise-bound}
\end{equation}
where $p=M/n_{c}^{2}$.
\end{lem}

As we focus on not too small $n_{c}$, which is chosen to be such
that $M\approx Cn_{c}(\log n_{c})^{2}$, the bound (\ref{eq:matrix-completion-noise-bound})
can be simplified as 
\begin{align*}
\delta & \leq4\sqrt{\frac{(2+Cn_{c}(\log n_{c})^{2}/n_{c}^{2})n_{c}}{Cn_{c}(\log n_{c})^{2}/n_{c}^{2}}}\epsilon+2\epsilon\\
 & \approx\sqrt{\frac{32}{C}}\frac{n_{c}}{\log n_{c}}\epsilon.
\end{align*}

Let $\mathbf{u}_{1}$ and $\hat{\mathbf{u}}_{1}=\mathbf{u}_{1}+\mathbf{e}_{1}$
be the dominant left singular vectors of $\mathbf{H}$ and $\hat{\mathbf{H}}_{c}$,
respectively. We exploit the following classical result from singular
vector perturbation analysis. 
\begin{lem}
[Singular vector perturbation \cite{vu:J11singular}]\label{lem:Singular-vector-perturbation}
Let $\sigma_{1}$ and $\sigma_{2}$ be the first and second dominant
singular values of $\mathbf{H}$. Then, 
\[
\sin\angle(\mathbf{u}_{1},\hat{\mathbf{u}}_{1})\leq\frac{2\|\hat{\mathbf{H}}_{\text{c}}-\mathbf{H}\|_{\text{F}}}{\sigma_{1}-\sigma_{2}}=\frac{2\delta}{\sigma_{1}-\sigma_{2}}.
\]
\end{lem}

By exploiting Lemma \ref{lem:Singular-vector-perturbation} for our
case, we have 
\begin{align*}
\sin\angle(\mathbf{u}_{1},\hat{\mathbf{u}}_{1}) & =\sqrt{1-\big|\mathbf{u}_{1}^{\text{T}}(\mathbf{u}_{1}+\mathbf{e}_{1})\big|^{2}}\\
 & =\sqrt{-2\mathbf{u}_{1}^{\text{T}}\mathbf{e}_{1}+\big|\mathbf{u}_{1}^{\text{T}}\mathbf{e}_{1}\big|^{2}}\\
 & \approx\sqrt{2\big|\mathbf{u}_{1}^{\text{T}}\mathbf{e}_{1}\big|}
\end{align*}
where $|\cdot|$ denotes the absolute value operator, and we drop
the second order term $|\mathbf{u}_{1}^{\text{T}}\mathbf{e}_{1}|^{2}$,
since $|\mathbf{u}_{1}^{\text{T}}\mathbf{e}_{1}|$ is small as we
focus on large $n_{c}(M)$. We also note that $\mathbf{u}_{1}^{\text{T}}\mathbf{e}_{1}\leq0$. 

Consider that we have chosen $M\approx Cn_{c}(\log n_{c})^{2}$, and
moreover, $\mathbf{H}$ is a rank-1 matrix with singular value $\sigma_{1}=\alpha$.
As a result, 
\begin{align*}
2\big|\mathbf{u}_{1}^{\text{T}}\mathbf{e}_{1}\big|\approx\sin^{2}\angle(\mathbf{u}_{1},\hat{\mathbf{u}}_{1})\leq\Big(\frac{2\delta}{\alpha}\Big)^{2} & \leq\frac{4}{\alpha^{2}}\frac{32}{C}\Big(\frac{n_{c}}{\log n_{c}}\Big)^{2}\bar{\epsilon}^{2}\\
 & =128u(0)^{2}K_{u}^{2}L^{6}n_{c}^{-3}.
\end{align*}

\subsubsection{Estimator based on Reflected Correlation}

Let $e(x)=\hat{u}(x)-u_{1}(x)$.  Define a reflected correlation function
as 
\[
R(t;x_{1}^{\text{S}})=\int_{-\infty}^{\infty}u(x-x_{1}^{\text{S}})u(-x-x_{1}^{\text{S}}+t)dx.
\]
Then, it follows that $R(t;x_{1}^{\text{S}})=\tau(2x_{1}^{\text{S}}-t)$.
As a result, we have 
\begin{align}
 & \hat{R}(t;\hat{\mathbf{u}}_{1})\nonumber \\
 & \quad=\int_{-\infty}^{\infty}\Big(u_{1}(x)+e(x)\Big)\Big(u_{1}(-x+t)+e(-x+t)\Big)dx\nonumber \\
 & \quad=\int_{-\infty}^{\infty}u_{1}(x)u_{1}(-x+t)dx+\int_{-\infty}^{\infty}u_{1}(x)e(-x+t)dx\nonumber \\
 & \qquad\qquad\int_{-\infty}^{\infty}e(x)u_{1}(-x+t)dx+\int_{-\infty}^{\infty}e(x)e(-x+t)dx\nonumber \\
 & \quad\approx R(t;x_{1}^{\text{S}})+\int_{-\infty}^{\infty}u_{1}(x)e(-x+t)dx\nonumber \\
 & \qquad\qquad\qquad+\int_{+\infty}^{-\infty}e(-y+t)u_{1}(y)(-dy)\label{eq:app-approx-a}\\
 & \quad=R(t;x_{1}^{\text{S}})+2\int_{-\infty}^{\infty}u_{1}(x)e(-x+t)dx\nonumber \\
 & \quad\approx R(t;x_{1}^{\text{S}})+2\mathbf{u}_{1}^{\text{T}}\mathbf{e}_{1}^{-t}\label{eq:app-approx-b}
\end{align}
where the first approximation (\ref{eq:app-approx-a}) is by dropping
the second order term $\int_{-\infty}^{\infty}e(x)e(-x+t)dx$, and
the second approximation (\ref{eq:app-approx-b}) is to use the inner
product $\mathbf{u}_{1}^{\text{T}}\mathbf{e}_{1}^{-t}$ to approximate
the integral based on assumptions A1 and A2 in Section \ref{sec:single-source}.
As a result, we have $R(t;x_{1}^{\text{S}})-\hat{R}(t;\hat{\mathbf{u}}_{1})\approx-2\mathbf{u}_{1}^{\text{T}}\mathbf{e}_{1}^{-t}$.

Recall that $\hat{t}=2\hat{x}_{1}^{\text{S}}$ maximizes $\hat{R}(\hat{t};\hat{\mathbf{u}}_{1})$
and $t^{*}=2x_{1}^{\text{S}}$ maximizes $R(t^{*};x_{1}^{\text{S}})=\tau(2x_{1}^{\text{S}}-t^{*})$.
We have 
\begin{align*}
 & \tau(0)-\tau\big(2\big|\hat{x}_{1}^{\text{S}}-x_{1}^{\text{S}}\big|\big)\\
 & \qquad=R(t^{*};x_{1}^{\text{S}})-\hat{R}(\hat{t};\hat{\mathbf{u}}_{1})\\
 & \qquad\approx-2\mathbf{u}_{1}^{\text{T}}\mathbf{e}_{1}^{-t}\\
 & \qquad\leq C_{e}2\big|\mathbf{u}_{1}^{\text{T}}\mathbf{e}_{1}\big|\\
 & \qquad\leq\mu_{u}L^{6}n^{-3}+o(n_{c}^{-3})
\end{align*}
where $\mu_{u}=C_{e}128u(0)^{2}K_{u}^{2}$ and $o(n_{c}^{-3})$ is
due to the fact that we keep omitting the higher order terms. Finally,
we obtain 
\[
\tau\big(2\big|\hat{x}_{1}^{\text{S}}-x_{1}^{\text{S}}\big|\big)=1-\mu_{u}L^{6}n^{-3}+o(n_{c}^{-3})
\]
and hence,
\[
\big|\hat{x}_{1}^{\text{S}}-x_{1}^{\text{S}}\big|\leq\frac{1}{2}\tau^{-1}(1-\mu_{u}L^{6}n_{c}^{-3}+o(n_{c}^{-3})).
\]

\subsection{Proof of Theorem \ref{thm:Unique-local-maximum}}

We first study the singular vectors in double source case. 
\begin{lem}
[Singular vectors in two source case]\label{lem:Eigenvectors} Let
$\mathbf{u}_{k}(\theta)$ and $\mathbf{v}_{k}(\theta)$ be the vectors
defined following (\ref{eq:uk}) and (\ref{eq:vk}) in the rotated
coordinate system $\mathcal{C}_{\theta}$. The \ac{svd} of $\mathbf{H}_{\theta}$
is given by 
\begin{equation}
\mathbf{H}_{\theta}=\alpha_{1}\mathbf{p}_{1}\mathbf{q}_{1}^{\text{T}}+\alpha_{2}\mathbf{p}_{2}\mathbf{q}_{2}^{\text{T}}\label{eq:H-theta-SVD}
\end{equation}
where $\alpha_{1}=\frac{\alpha}{2}\|\mathbf{u}_{1}+\mathbf{u}_{2}\|\|\mathbf{v}_{1}+\mathbf{v}_{2}\|$
and $\alpha_{2}=\frac{\alpha}{2}\|\mathbf{u}_{1}-\mathbf{u}_{2}\|\|\mathbf{v}_{1}-\mathbf{v}_{2}\|$
are the singular values, and
\[
\mathbf{p}_{1}=\frac{\mathbf{u}_{1}+\mathbf{u}_{2}}{\|\mathbf{u}_{1}+\mathbf{u}_{2}\|},\quad\mathbf{q}_{1}=\frac{\mathbf{v}_{1}+\mathbf{v}_{2}}{\|\mathbf{v}_{1}+\mathbf{v}_{2}\|}
\]
\[
\mathbf{p}_{2}=\frac{\mathbf{u}_{1}-\mathbf{u}_{2}}{\|\mathbf{u}_{1}-\mathbf{u}_{2}\|},\quad\mathbf{q}_{2}=\frac{\mathbf{v}_{1}-\mathbf{v}_{2}}{\|\mathbf{v}_{1}-\mathbf{v}_{2}\|}
\]
are the corresponding singular vectors. 
\end{lem}
\begin{proof}
First, 
\begin{align*}
\mathbf{H}_{\theta} & =\bar{\alpha}\big(\mathbf{u}_{1}\mathbf{v}_{1}^{\text{T}}+\mathbf{u}_{2}\mathbf{v}_{2}^{\text{T}}\big)\\
 & =\frac{\bar{\alpha}}{2}\Big[(\mathbf{u}_{1}+\mathbf{u}_{2})(\mathbf{v}_{1}+\mathbf{v}_{2})^{\text{T}}+(\mathbf{u}_{1}-\mathbf{u}_{2})(\mathbf{v}_{1}-\mathbf{v}_{2})^{\text{T}}\Big]\\
 & =\alpha_{1}\mathbf{p}_{1}\mathbf{q}_{1}^{\text{T}}+\alpha_{2}\mathbf{p}_{2}\mathbf{q}_{2}^{\text{T}}
\end{align*}
Hence, these four vectors form a decomposition of $\mathbf{H}_{\theta}$. 

Second, we have 
\begin{align*}
\mathbf{p}_{1}^{\text{T}}\mathbf{p}_{2} & =c(\mathbf{u}_{1}+\mathbf{u}_{2})^{\text{T}}(\mathbf{u}_{1}-\mathbf{u}_{2})\\
 & =c\big(\|\mathbf{u}_{1}\|^{2}-\|\mathbf{u}_{2}\|^{2}\big)\\
 & =0
\end{align*}
where $c=1/(\|\mathbf{u}_{1}+\mathbf{u}_{2}\|\|\mathbf{u}_{1}-\mathbf{u}_{2}\|)$.
Similarly, $\mathbf{q}_{1}^{\text{T}}\mathbf{q}_{2}=0$. In addition,
all the four vectors have unit norm. 

As a result, (\ref{eq:H-theta-SVD}) is the \ac{svd} of $\mathbf{H}_{\theta}$.
\end{proof}

Consider an arbitrary coordinate system. \Ac{wlog} (due to Assumption
1), assume that the first source is located at the origin, $x_{1}^{\text{S}}=0$
and $y_{1}^{\text{S}}=0$, and the second source is away from the
first source with distance $D$ and angle $\theta$ to the $x$-axis,
$x_{2}^{\text{S}}=D\cos\theta$ and $y_{2}^{\text{S}}=D\sin\theta$.
In addition, defining 
\[
u_{\text{c}}(x,\theta)\triangleq u(x-D\cos\theta),\qquad u_{\text{s}}(x,\theta)\triangleq u(x-D\sin\theta)
\]
we have 
\begin{align*}
\mathbf{u}_{1} & =\sqrt{\delta}\big[u(x_{1}),u(x_{2}),\dots,u(x_{N})\big]^{\text{T}}\\
\mathbf{v}_{1} & =\sqrt{\delta}\big[u(y_{1}),u(y_{2}),\dots,u(y_{M})\big]^{\text{T}}\\
\mathbf{u}_{2} & =\sqrt{\delta}\big[u_{\text{c}}(x_{1},\theta),u_{\text{c}}(x_{2},\theta),\dots,u_{\text{c}}(x_{N},\theta)\big]^{\text{T}}\\
\mathbf{v}_{2} & =\sqrt{\delta}\big[u_{\text{s}}(y_{1},\theta),u_{\text{s}}(y_{2},\theta),\dots,u_{\text{s}}(y_{M},\theta)\big]^{\text{T}}.
\end{align*}

Based on assumption A1 and A2, we have 
\begin{align}
\|\mathbf{u}_{k}\|^{2}=\Big(\frac{L}{n}\Big)^{2}\sum_{i=1}^{N}u(x_{i}-x_{k}^{\text{S}})^{2} & \approx\int_{x_{1}}^{x_{n-1}}u(x-x_{k}^{\text{S}})^{2}dx\nonumber \\
 & \approx\int_{-\infty}^{\infty}u(x-x_{k}^{\text{S}})^{2}dx=1\label{eq:approximation-integral-1}
\end{align}
and similar integrals apply to $\mathbf{v}_{k}$. 

As an equivalent statement to Theorem \ref{thm:Unique-local-maximum},
we need to show that $\rho(\theta)$ is a strictly increasing function
in $\theta\in(0,\frac{\pi}{4})$. Equivalently, we should prove that
the function 
\begin{align*}
 & \frac{\lambda_{2}(\mathbf{H}_{\theta})^{2}}{\lambda_{1}(\mathbf{H}_{\theta})^{2}}\\
 & \quad\approx\frac{\int_{-\infty}^{\infty}\big(u(x)-u_{\text{c}}(x,\theta)\big)^{2}dx}{\int_{-\infty}^{\infty}\big(u(x)+u_{\text{c}}(x,\theta)\big)^{2}dx}\frac{\int_{-\infty}^{\infty}\big(u(x)-u_{\text{s}}(x,\theta)\big)^{2}dx}{\int_{-\infty}^{\infty}\big(u(x)+u_{\text{s}}(x,\theta)\big)^{2}dx}\\
 & \quad\triangleq\mu(\theta)
\end{align*}
is strictly increasing in $\theta\in(0,\frac{\pi}{4})$, where the
approximated integrals are obtained from (\ref{eq:approximation-integral-1}).

To simplify the notation, define the integration operator $\left\langle \cdot\right\rangle $
as 
\[
\left\langle f\right\rangle \triangleq\int_{-\infty}^{\infty}f(x,\theta)dx
\]
for a function $f(x,\theta)$. By definition, the integration operator
is linear and satisfies the additive property, i.e., $\langle af\rangle=a\langle f\rangle$
and $\langle f+g\rangle=\langle f\rangle+\langle g\rangle$, for a
constant $a$ and a function $g(x,\theta)$. As a result, $\langle(u-u_{\text{c}})^{2}\rangle=\langle u^{2}\rangle+\langle u_{\text{c}}^{2}\rangle-2\langle u\cdot u_{\text{c}}\rangle=2\big(1-\langle u\cdot u_{\text{c}}\rangle\big)$,
and the function $\mu(\theta)$ can be written as 
\begin{equation}
\mu(\theta)=\frac{\big(1-\langle u\cdot u_{\text{c}}\rangle\big)\big(1-\langle u\cdot u_{\text{s}}\rangle\big)}{\big(1+\langle u\cdot u_{\text{c}}\rangle\big)\big(1+\langle u\cdot u_{\text{s}}\rangle\big)}.\label{eq:mu-function}
\end{equation}

In addition, from the properties in calculus, if $f(x,\theta)$ and
$\frac{\partial}{\partial\theta}f(x,\theta)$ are continuous in $\theta,$
then 
\begin{align*}
\frac{d}{d\theta}\left\langle f\right\rangle  & =\frac{d}{d\theta}\int_{-\infty}^{\infty}f(x,\theta)dx\\
 & =\int_{-\infty}^{\infty}\frac{\partial}{\partial\theta}f(x,\theta)dx=\Big\langle\frac{\partial}{\partial\theta}f\Big\rangle.
\end{align*}
Therefore, defining 
\begin{align*}
u_{\text{c}}^{'}(x,\theta) & \triangleq\frac{d}{dx}u(x)\big|_{x=x-D\cos\theta}\\
u_{\text{s}}^{'}(x,\theta) & \triangleq\frac{d}{dx}u(x)\big|_{x=x-D\sin\theta}
\end{align*}
we have
\begin{align*}
\frac{d}{d\theta}\langle u\cdot u_{\text{c}}\rangle & =\langle u\cdot\frac{\partial}{\partial\theta}u_{\text{c}}(x,\theta)\rangle=\langle u\cdot u_{\text{c}}^{'}\rangle D\sin\theta\\
\frac{d}{d\theta}\langle u\cdot u_{\text{s}}\rangle & =\langle u\cdot\frac{\partial}{\partial\theta}u_{\text{s}}(x,\theta)\rangle=-\langle u\cdot u_{\text{s}}^{'}\rangle D\cos\theta.
\end{align*}

With some algebra, the derivative of $\mu(\theta)$ can be obtained
as 
\begin{align*}
\frac{d}{d\theta}\mu(\theta) & =\eta\Big[D\cos\theta\langle u\cdot u_{\text{s}}^{'}\rangle\big(1-\langle u\cdot u_{\text{c}}\rangle^{2}\big)\\
 & \qquad\qquad-D\sin\theta\langle u\cdot u_{\text{c}}^{'}\rangle\big(1-\langle u\cdot u_{\text{s}}\rangle^{2}\big)\Big]\\
 & =\eta\Big[-t\cdot\tau^{'}(s)\big(1-\tau(t)^{2}\big)+s\cdot\tau^{'}(t)\big(1-\tau(s)^{2}\big)\Big]
\end{align*}
where $\eta=2\big(1+\langle u\cdot u_{\text{c}}\rangle\big)^{-2}\big(1+\langle u\cdot u_{\text{s}}\rangle\big)^{-2}$,
$t=D\cos\theta$, and $s=D\sin\theta$. 

Note that $0<s<t$ for $0<\theta<\frac{\pi}{4}$. Applying condition
(\ref{eq:correlation-condition}), we have 
\begin{align*}
\frac{d}{d\theta}\mu(\theta) & >\eta\cdot t\cdot\tau^{'}(s)\Big[\big(1-\tau(s)^{2}\big)-\big(1-\tau(t)^{2}\big)\Big]\\
 & =\eta\cdot t\cdot\tau^{'}(s)\big(\tau(t)^{2}-\tau(s)^{2}\big)\\
 & >0
\end{align*}
since $\tau^{'}(s)<0$ and $\tau(t)<\tau(s)$ for $0<s<t$.

This confirms that $\mu(\theta)$ is a strictly increasing function,
and hence $\rho(\theta)$ is a strictly increasing function in $\theta\in(0,\frac{\pi}{4})$.
The results in Theorem \ref{thm:Unique-local-maximum} is confirmed. 

\bibliographystyle{IEEEtran}

\begin{thebibliography}{10}
\providecommand{\url}[1]{#1}
\csname url@samestyle\endcsname
\providecommand{\newblock}{\relax}
\providecommand{\bibinfo}[2]{#2}
\providecommand{\BIBentrySTDinterwordspacing}{\spaceskip=0pt\relax}
\providecommand{\BIBentryALTinterwordstretchfactor}{4}
\providecommand{\BIBentryALTinterwordspacing}{\spaceskip=\fontdimen2\font plus
\BIBentryALTinterwordstretchfactor\fontdimen3\font minus
  \fontdimen4\font\relax}
\providecommand{\BIBforeignlanguage}[2]{{%
\expandafter\ifx\csname l@#1\endcsname\relax
\typeout{** WARNING: IEEEtran.bst: No hyphenation pattern has been}%
\typeout{** loaded for the language `#1'. Using the pattern for}%
\typeout{** the default language instead.}%
\else
\language=\csname l@#1\endcsname
\fi
#2}}
\providecommand{\BIBdecl}{\relax}
\BIBdecl

\bibitem{BecStoLi:J08}
A.~Beck, P.~Stoica, and J.~Li, ``Exact and approximate solutions of source
  localization problems,'' \emph{{IEEE} Trans. Signal Process.}, vol.~56,
  no.~5, pp. 1770--1778, 2008.

\bibitem{QiXiuYua:J13}
H.-D. Qi, N.~Xiu, and X.~Yuan, ``A lagrangian dual approach to the
  single-source localization problem,'' \emph{{IEEE} Trans. Signal Process.},
  vol.~61, no.~15, pp. 3815--3826, 2013.

\bibitem{SheHu:J05}
X.~Sheng and Y.-H. Hu, ``Maximum likelihood multiple-source localization using
  acoustic energy measurements with wireless sensor networks,'' \emph{{IEEE}
  Trans. Signal Process.}, vol.~53, no.~1, pp. 44--53, 2005.

\bibitem{MeeMitNar:J08}
C.~Meesookho, U.~Mitra, and S.~Narayanan, ``On energy-based acoustic source
  localization for sensor networks,'' \emph{{IEEE} Trans. Signal Process.},
  vol.~56, no.~1, pp. 365--377, 2008.

\bibitem{LiuHuPan:J12}
Y.~Liu, Y.~H. Hu, and Q.~Pan, ``Distributed, robust acoustic source
  localization in a wireless sensor network,'' \emph{{IEEE} Trans. Signal
  Process.}, vol.~60, no.~8, pp. 4350--4359, 2012.

\bibitem{ZisWax:C88}
I.~Ziskind and M.~Wax, ``Maximum likelihood localization of multiple sources by
  alternating projection,'' \emph{Proc. IEEE Int. Conf. Acoustics, Speech, and
  Signal Processing}, vol.~36, no.~10, pp. 1553--1560, 1988.

\bibitem{LefRea:J17}
R.~Lefort, G.~Real, and A.~Dr{\'e}meau, ``Direct regressions for underwater
  acoustic source localization in fluctuating oceans,'' \emph{Applied
  Acoustics}, vol. 116, pp. 303--310, 2017.

\bibitem{NguJorSin:J05}
X.~Nguyen, M.~I. Jordan, and B.~Sinopoli, ``A kernel-based learning approach to
  ad hoc sensor network localization,'' \emph{ACM Trans. on Sensor Networks},
  vol.~1, no.~1, pp. 134--152, 2005.

\bibitem{JinSohWon:J10}
Y.~Jin, W.-S. Soh, and W.-C. Wong, ``Indoor localization with channel impulse
  response based fingerprint and nonparametric regression,'' \emph{{IEEE}
  Trans. Wireless Commun.}, vol.~9, no.~3, pp. 1120--1127, 2010.

\bibitem{KimParYooKimPar:J13}
W.~Kim, J.~Park, J.~Yoo, H.~J. Kim, and C.~G. Park, ``Target localization using
  ensemble support vector regression in wireless sensor networks,'' \emph{IEEE
  Trans. on Cybernetics}, vol.~43, no.~4, pp. 1189--1198, 2013.

\bibitem{ChoMit:C15}
S.~Choudhary and U.~Mitra, ``Analysis of target detection via matrix
  completion,'' in \emph{Proc. IEEE Int. Conf. Acoustics, Speech, and Signal
  Processing}, 2015, pp. 3771--3775.

\bibitem{ChoKumNarMit:J16}
S.~Choudhary, N.~Kumar, S.~Narayanan, and U.~Mitra, ``Active target
  localization using low-rank matrix completion and unimodal regression,''
  \emph{arXiv preprint arXiv:1601.07254}, 2016.

\bibitem{CanPla:J10}
E.~J. Candes and Y.~Plan, ``Matrix completion with noise,'' \emph{Proceedings
  of the IEEE}, vol.~98, no.~6, pp. 925--936, 2010.

\bibitem{peterson:b95}
R.~Peterson, R.~Ziemer, and D.~Borth, \emph{Introduction to spread spectrum
  systems}.\hskip 1em plus 0.5em minus 0.4em\relax Englewood Cliffs, NJ:
  Prentice-Hall, 1995.

\bibitem{vu:J11singular}
V.~Vu, ``Singular vectors under random perturbation,'' \emph{Random Structures
  \& Algorithms}, vol.~39, no.~4, pp. 526--538, 2011.

\end{thebibliography}

\end{document}